\documentclass{article}

\usepackage[english]{babel}
\usepackage{natbib}
\usepackage[affil-it]{authblk}

\usepackage[letterpaper,top=2cm,bottom=2cm,left=3cm,right=3cm,marginparwidth=1.75cm]{geometry}

\usepackage{amsmath,amsthm}
\usepackage{graphicx}
\usepackage{xcolor}
\usepackage{booktabs}
\usepackage[colorlinks=true, allcolors=blue]{hyperref}

\newtheorem{proposition}{Proposition}
\newtheorem{lemma}{Lemma}
\newtheorem{corollary}{Corollary}
\newtheorem{definition}{Definition}
\newtheorem{theorem}{Theorem}

\newtheorem{problem}{Problem}
\newtheorem{question}{Question}

\title{Arrow's single peaked domains, richness, and domains for plurality and the Borda count}

\author[1]{Klas Markstr{\"o}m\thanks{Corresponding author: klas.markstrom@umu.se}}
\author[2]{S\o ren Riis}
\author[2]{Bei Zhou\footnote{Authors are  listed in alphabetical order}}
\affil[1]{Ume\aa\ University} 
\affil[2]{Queen Mary University of London}

\date{}

\begin{document}
\maketitle
\begin{abstract}
In this paper  we extend the study of Arrow's generalisation of Black's single-peaked domain  and connect this to domains where voting rules satisfy different versions of independence of irrelevant alternatives.

First we report on a computational generation of all non-isomorphic Arrow's single-peaked domains on $n\leq 9$ alternatives.  Next, we  introduce a quantitative measure of richness for domains, as the largest number $r$ such that every alternative is given every rank between 1 and $r$ by the orders in the domain.  We investigate the richness of Arrow's single-peaked domains and prove that Black's single-peaked domain has the highest possible richness, but it is not the only domain which attains the maximum.

After this we connect Arrow's single-peaked domains to the discussion by Dasgupta, Maskin and others of domains on which plurality and the Borda count satisfy different versions of Independence of Irrelevant alternatives (IIA). For Nash's version of IIA and plurality, it turns out the domains are exactly the duals of Arrow's single-peaked domains. As a consequence there can be at most two alternatives which are ranked first in any such domain.

For the Borda count  both Arrow's and Nash's versions of  IIA lead to a   maximum domain size which is exponentially  smaller than $2^{n-1}$, the size of Black's single-peaked domain. 
\end{abstract}

\section{Introduction}

Ever since the foundational work by Condorcet \cite{Con85} it has been known that simple majority voting can display various unwanted effects if the voter's opinions are diverse enough. In contrast,  \cite{black} showed that when the voters rank candidates according to their position on a common political axis, a majority decision will always lead to a transitive ranking of the candidates.
The maximum domain of this type is now known as Black's single peaked domain and it is known that up to isomorphism there is a unique such domain, of size $2^{n-1}$, for each number of voters $n$.  In addition to Black's original definition there are now several structural characterisations of this domain, see \cite{PUPPE201855}. Strategy-proof voting on single-peaked domains is well understood \cite{Mou80}  and more generally it is known that many of the hard computational problems of social choice become much simpler on single-peaked domains, see \cite{elkind2022preference}. 

In terms of modelling real-life preferences, for example in general elections, the assumption of a single-peaked domain has been debated. At some periods in time, in certain countries, a large part of the population has based their preferences on a left-right political axis, however even in such circumstances one could expect there to be outliers which make the domain deviate from a single-peaked one. Some authors \cite{list2013deliberation} have investigated how e.g. deliberation may lead to a domain which is closer to being single-peaked, while others \cite{egan2014something} have studied situations which may lead to a domain with two peaks, on a common political axis. A concrete example deviating from single-peakedness given in \cite{dasgupta2008robustness} is the French presidential election of 2002.  Due to these deviations, one could, even in situations which favour single-peakedness, expect to see a domain which has a large intersection with Black's single-peaked domain, rather than an actual subdomain of it.  

One class of such domains was introduced by Arrow, who in  \cite{arrow63} introduced  a weaker domain condition than Black's, requiring only that the restriction of the domain to any triple of alternatives is single-peaked in Black's sense, along an axis which may change between triples. Arrow showed that  domains satisfying this weaker condition will be Condorcet domains, where pairwise majorities lead to a transitive ranking,  just like Black's domain. Following Raynaud \cite{raynaud1981paradoxical} this class is now called Arrow's  single-peaked domains. Inada  \cite{Inada64}  showed that this class is more general than Black's, in that there are additional maximal Arrow's single-peaked domains, apart from Black's original domain.

The structure of Arrow's single-peaked domains has been studied by several authors after Inada. \cite{romero1978variations} constructed an infinite sequence of single-peaked domains, distinct from Black's.  \cite{raynaud1981paradoxical} showed that all Arrow's single-peaked domains have size at most $2^{n-1}$, just like Black's original domain. These works were unified by  \cite{SLINKO2019166} who showed that all maximal Arrow's single-peaked domains have size exactly $2^{n-1}$ and gave a recursive description of the structure of all domains in the class.  From the latter it also follows  that all such domains have at most two alternatives which can be ranked last in any order from the domain, and that every alternative is ranked first by some order,  i.e. these domains are \emph{minimally rich}, as defined by \cite{Aswal03}. Slinko also enumerated the non-isomorphic Arrow's single-peaked domains on 4 and 5 alternatives respectively,  a result which was later extended up to 7 alternatives in \cite{akello2023condorcet} and $n=8$ in \cite{liversidge2020counting}.  

In this paper we will continue the study of Arrow's single-peaked domains in several ways. First we report the result of a computational search for non-isomorphic Arrow's single-peaked domains. The result of this search verifies the existing counts for such domains up to $n=8$ and extends this to $n=9$.  In terms of many structural properties all Arrow's single-peaked domains are equivalent, for example they all have the same local diversity as defined in \cite{abundance}, and  they are all minimally rich domains.  In Section \ref{sec:rich} we extend the property of being minimally rich to a quantitative version, where we count the number of initial ranks at which every alternative can appear. Having richness at least 1 corresponds to being minimally rich and as we shall see black's single-peaked domain provides an example obtaining the  largest possible richness for an Arrow's single-peaked domain, while not being the only Arrow's single-peaked  domain with that richness. 

In Section \ref{sec:systems} we first connect Arrow's single-peaked domains to the analysis by Dasgupta and Maskin \cite{dasgupta2008robustness} of domains on which various voting rules are well-behaved. Here the focus is on rules which pick a single winner, rather than giving a ranking of all candidates.  We show that the domains on which the plurality rule satisfies  Nash's version \cite{Nash50} of independence of irrelevant alternatives  are duals of  Arrow's single-peaked domains. These domains can be seen as an Arrow type generalisation of the single-dipped domain, just as the Arrow's single-peaked domains generalise Black's single-peaked domain. As a consequence we find that, in any domain where more than two alternatives can be ranked first, the plurality rule will fail to satisfy  Nash's  version of independence of irrelevant alternatives.  

Following this we explicate results by \cite{barbie2006non}  and \cite{dasgupta2008robustness} by identifying the maximum domains on which the Borda count satisfies either Arrow's, or Nash's versions of independence of irrelevant alternatives and compare their domain sizes and number of possible winners.  Finally we give a short discussion of our results.

\section{Notation and definitions}

Let a finite set  $X_n=[n]=\{1, \dots, n\}$ be the set of alternatives. Let $L(X_n)$ be the set of all linear orders over $X$.  Each agent $i \in N$ has a preference order $P_i$ over $X_n$ (each preference order is a linear order). 

A subset of preference orders $D\subseteq L(X)$ is called a \emph{domain} of preference orders. A domain $D$ is a \emph{Condorcet domain} if whenever the preferences of all agents belong to the domain, the majority relation of any preference profile with an odd number of agents is transitive. A Condorcet domain $D$ is \emph{maximal} if every Condorcet domain $D'\supset D$ (on the same set of alternatives) coincides with $D$. A Condorcet domain $D$ is \emph{unitary} if it contains the order $12\ldots n$. Every Condorcet domain can be made unitary by renaming alternatives. However, this should not automatically be seen as introducing a meaningful societal axis.  We will mainly consider unitary domains.

The \emph{restriction} of a domain $D$ to a subset $A\subset X$ is the set of linear orders from $L(A)$ obtained by restricting each linear order from $D$ to $A$. For set complements we  will use the notation  $\overline{A}=[n]\setminus A$, where $A\subseteq [n]$.

The \emph{dual} of a domain $D$ is the domain obtained by reversing each linear order in $D$.

Following \cite{akello2023condorcet} we say that an alternative $x$ is a \emph{fixed point} of a domain $D$ if every order in $D$ assigns the same rank to $x$.

Sen \cite{Sen1966} proved that for each Condorcet domain, restriction of the domain to each triple $(a,b,c)$ of alternatives satisfies a never condition of one of the forms: $x$ not top, $x$ not middle, $x$ not bottom/last, for some alternative $x$ in the triple.   Here top, middle, and bottom/last refer to being ranked first, second, and third respectively.

For a unitary domain one can describe never conditions in a form given by Fishburn; $iNj$, $i,j\in [3]$. $iNj$ means that $i^{th}$ alternative from the triple according to the global ascending order does not get ranked $j^{th}$ within this triple in any order from the domain. Unitary domains have at most six types of never conditions. 

A domain $D$ is single-peaked, in the sense of Black, if there exists an order $>$ in $D$  such that for every other order $\succ$ in $D$ and every alternative $x$  the set $\{y:y \succ x\}$ forms an interval in $>$.  The order $>$ is often referred to as either the axis or the spectrum of the domain.

A domain is an Arrow's single-peaked domain if and only if every triple of alternatives satisfies a never bottom condition, or in the Fishburn version a never condition of the form $xN3$ for some $x$ from the triple. For a unitary domain this means that every triple satisfies either $1N3$ or $2N3$.  

We will say that a domain is an Arrow's single-dipped domain if it is the dual of an Arrow's single-peaked domain, or equivalently if every triple satisfies a never top condition. 

A \emph{profile}  is a vector $\mathbf{R}=(R_1,R_2,\ldots,R_N)$ such that $R_i$ is the preferences of voter/agent $i$. 

For many of our computational results  we have used the CDL library \cite{zhou2023cdl}, which provides functions for the study of Condorcet domains and other types of restricted preference orders. Whenever possible we have also verified those results by using Mathematica. 

\section{Enumeration of Arrow's single-peaked domains}\label{sec:gen}
In order to extend the existing enumerations of Arrow's single-peaked domains and get a complete data set to make available online we have implemented two variations of a search algorithm using the functionalities provided by the  CDL package.  

The implementations made here are completely independent of the codes and methods used in \cite{akello2023condorcet} and \cite{liversidge2020counting}, thereby providing a three-way independent correctness check of the data up to $n=7$, and a two-way check for $n=8$. 

All the domains generated here are publicly available for download at \cite{Web1}.

\subsection{The search algorithm}
The CDL library \cite{zhou2023cdl} provides basic functionality for implementing searches for different types of domains.    In our basic set up we do not work directly with  a domain but instead with never conditions which define it.   So, we represent a (partial) domain by a string of length ${n \choose 3}$, one symbol per triple of alternatives and as symbols we use the two possible Fishburn type never conditions and an empty symbol. Each Arrow's single-peaked domain then corresponds to a string with no empty symbols.

We first implemented a simple depth-first search algorithm, which starts with a string of only empty symbols and then replaces these with never conditions.  As in \cite{zhou2023new} a database of restrictions to a fixed number of alternatives is used to reject those strings which cannot lead to maximal domains, e.g. because the partial domain is not copious. In the end, the set of generated domains was reduced to isomorphism class representatives using CDL.
This basic algorithm can generate all non-isomorphic Arrow's single-peaked domains for $n \leq 8$ alternatives in a reasonable amount of time. However, the number of CDs on nine alternatives is much larger and leaving the isomorphism reduction to the end becomes too inefficient, prompting us to design an improved algorithm. 

The second algorithm is still a depth-first search but instead of removing isomorphic versions of a domain at the end of the search, the second algorithm performs this reduction during the search. When an ordering of triples of alternatives, and an ordering of the set of never conditions together with the empty symbol, are fixed, the set of strings corresponding to (partial) domains can be sorted by comparing their corresponding list of never conditions lexicographically. We will now aim to keep the lexicographically smallest member of each isomorphism class. This in turn means that the empty symbol should be taken as the largest element in the order for the symbols.

We can now apply permutations of the names of the alternatives and map the string of a (partial) domain to the string of an isomorphic (partial) domain.  Given a permutation $g$ on the set of alternatives from the assigned triples, a triple $(i, j, k)$ with the never condition $xNp$  will be transformed to a new triple $(g(i), g(j), g(k))$ with the never condition $g(x)Np$. If there is no never condition for the first triple then the transformed triple will not have one either.   Applying the permutation $g$ to all triples and listing the new triples according to the predefined ordering of triples leads to a string which we can compare lexicographically with the original list of triples and rules. If the new string is lexicographically smaller than the original string, then the search terminates for the original string, since it  cannot lead to a lexicographically minimal string.  This new search algorithm removes a large part of the entire search space, and in the end, all the surviving CDs are non-isomorphic. 

Our new form of isomorphism rejection is general, meaning it can also be applied to some other types of domains, and will be included in the CDL package.

\subsection{Results}
Let $a(n)$ denote the number of non-isomorphic maximal Arrow's single-peaks domains on $n$ alternatives. For $n=3,\ldots,9$ we have
$$a(n)=1,2,6,40,560,17024, \bf{1066496}$$
Up to $n=7$ these numbers agree with the enumerations from \cite{SLINKO2019166} and \cite{akello2023condorcet}, and the software implementations are completely independent. For $n=8$ our value agrees with that given in the preprint \cite{liversidge2020counting}. The value for $n=9$ is new, and as we will see in the next section our data passes a non-trivial correctness check. 

\section{Richness in general domains and Arrow's single-peaked domains}\label{sec:rich}
Our next step is to provide a quantitative refinement of the property of being minimally rich.   A domain is minimally rich if every alternative is ranked first by some order from the domain. Here we now ask, for how many of the consecutive highest ranks do we see all alternatives?  More formally:
\begin{definition}
A domain $D$ is $k$-rich  if for every $k'\leq k$ and every alternative $j$ there is an order in $D$ which gives alternative  $j$ rank $k'$.  

A domain has richness $k$ if it is $k$-rich but not $(k+1)$-rich.  
\end{definition}

Being 1-rich is equivalent to being minimally rich.   For general domains, not just Arrow's single-peaked ones,  there are several natural interpretations of richness.  First, we can see this as  measure of the amount of variety of highest ranked alternatives. So a domain is 1-rich if every alternative can win an election on this domain, and with higher richness we also find that each alternative  can be ranked everywhere in an interval of top positions.  Second, we can view this as a robustness measure for the usual minimal richness, since a $k$-rich domain is guaranteed to remain minimally rich if the top ranked alternative in every order is removed.

The unrestricted domain, consisting of all linear orders on $X$, has richness $n$, the maximum possible.  A much smaller domain  with richness $n$ can be constructed by taking the $n$ cyclic shifts of a single order on $X$.  For Condorcet domains the largest possible richness as a function of $n$ is not known. Using the data from \cite{akello2023condorcet} we tested the richness of all non-isomorphic maximal Condorcet domains for $n \leq 6$. For $n=3$ the maximum is 2  and for $n=4,\ldots 6$ the maximum is 4. For $n=4$ richness 4 is achieved by one of the two symmetric Condorcet domains of size 8. It is natural to ask how this develops for larger $n$.
\begin{question}
    What is the maximum possible richness of a Condorcet domain on $n$ alternatives?
\end{question}

We have computed the richness for all Arrow's single-peaked domains up to $n=9$, and display the results in Table \ref{tab:rich}.  In our data we observe that the smallest richness is 2 for every $n$  and that the maximum richness behaves like $\lfloor n/2 \rfloor+1$, for this range of values of $n$. We also note that the number of non-isomorphic domains of richness 2 is exactly equal to $a(n-1)$.

\begin{table}[h]
\centering
\setlength{\tabcolsep}{9pt}
\begin{tabular}{*{6}{c}}
\toprule
& \multicolumn{4}{c}{\textbf{Richness}} \\
\cmidrule(lr){2-5}
\textbf{n} & 2 & 3 & 4 & 5    \\
\midrule
3 & 1 &  - & -  & -  \\
4 &  1 & 1  & -  & -  \\
5 &  2 & 4  & -  &  -  \\
6 &  6  & 31   & 3  & - \\
7 &  40 & 439  &  81  & -  \\
8 & 560  &  12327    & 4101    & 36   \\
9 & 17024  & 696497  & 346635  & 6340   \\
\bottomrule
\end{tabular}
\caption{Richness for Arrow's single-peaked domains}\label{tab:rich}
\end{table}

\cite{SLINKO2019166} showed that all Arrow's single-peaked domains are minimally rich and our first result strengthens this slightly. 
\begin{lemma}\label{lem:2rich}
    All maximal Arrow's single-peaked domains are 2-rich.
\end{lemma}
\begin{proof}
    The statement is true for $n=3$ and $n=4$.
     
    Let $a_1$ and $a_2$ be  the two terminal alternatives of $D$ and $\Tilde{D_i}$ be the set of orders with $a_i$ as the terminal alternative.
     
    Let $D_i$ be the restriction of $\Tilde{D_i}$ to $A\setminus a_i$. By Lemma 4.3 of \cite{SLINKO2019166} each $D_i$ is a maximal Arrow's single-peaked domain, and by induction each $D_i$ is also 2-rich.

    Hence $\Tilde{D_i}$ has all the alternatives from $A\setminus a_i$ in the two highest ranks  and $a_i$ as the unique terminal alternative. In turn,  $D$ is the union of $\Tilde{D_1}$ and $\Tilde{D_2}$, so $D$ has all alternatives in rank 1 and 2. 
\end{proof}

We will refer to the set of ranks used by an alternative $a$ as the \emph{range} of $a$.  The range of an alternative in a maximal Arrow's single-peaked domains is strongly constrained.
\begin{proposition}\label{prop:int}
    The set of ranks used by an alternative in a maximal Arrow's single-peaked domain  is an interval starting at 1.
\end{proposition}
\begin{proof}
    Assume that $D$ is a maximal Arrow's single-peaked domain on $n$ alternatives and let $a_1$ and $a_2$ be the two terminal alternatives.

    The proposition is true for $n=3$. We will assume that it is true for less than $n$ alternatives.
 
    Assume that $a_2$ does not appear with rank $n-1$. For any other alternative $x$ the triple $(x,a_1,a_2)$ must have  $xN3$ as its never condition. Let $\sigma_1$ be an order from $D$ restricted  to $A\setminus\{a_1,a_2\}$. Let $\sigma_2$ be the order obtained from $\sigma_1$ by giving $a_2$ rank $n-1$ and $a_1$ rank $n$. Adding $\sigma_2$ to $D$ does not violate any never condition. So, the assumption that $a_2$ does not get rank $n-1$ contradicts the maximality of $D$.

    From this it follows that $a_1$ is a terminal alternative of $D_2$, and $a_2$ one for $D_1$. For each $D_i$ the proposition is true by induction.   Since $D$ is the union of $\tilde{D}_1$ and $\tilde{D}_2$ the set of ranks used by a non-terminal alternative $y$ is the union of two intervals, which is an interval. For the two terminal alternatives the set of ranks is the interval up to $n-1$ together with $n$, which is also an interval.
\end{proof}
For any alternative $a$ let $r_a$ denote the maximum value such that $a$ appears with all ranks in the interval $[1,\ldots,r_a]$. 

\cite{SLINKO2019166} showed that an Arrow's single-peaked domain has at most two terminal alternatives. That result is a special case of the following lemma.
\begin{lemma}\label{lem:pair}
    For each alternative $i$ there is at most one other alternative $j$  with $r_i=r_j$ in an Arrow's single-peaked domain.
\end{lemma}
\begin{proof}
    Assume that $i,j,k$ have the same range, then the triple $(i,j,k)$ cannot satisfy a never bottom condition, since they all appear with rank $r_i$. This contradicts that the domain is Arrow's single-peaked.   Hence there are at most two $i,j$ with $r_i=r_j$. 
\end{proof}
The result from  \cite{SLINKO2019166}  is obtained by letting $i$ be any terminal alternative of the domain.  The proposition and lemma together give the following:
\begin{corollary}
    {\ }
    \begin{enumerate}
        \item An Arrow's single-peaked domains with richness $r$ has at least one and at most two alternatives which are only given ranks 1 to $r$.
        
        \item The terminal alternatives are the only two alternatives which appear with all ranks.    
    \end{enumerate}
\end{corollary}

Lemma \ref{lem:2rich} showed that all Arrow's single-peaked domains are 2-rich. In the other direction we have an upper bound for the richness of all Arrow's single-peaked domains.
\begin{theorem}
    The richness of an Arrow's single-peaked domain on $n$ alternatives is at most $\lfloor n/2 \rfloor+1$.
\end{theorem}
\begin{proof}
It is easy to check that the statement is true for $n=3,4$. We will assume that it is true for up to $n-1$ alternatives.

Let  $i$ be an element of minimal richness $r_i$.  If the proposition does not hold for $n$ then $r_i\geq \lfloor n/2 \rfloor+2$.  Let $j,k$ be another two alternatives with as small richness as possible. By Lemma \ref{lem:pair} we either have $r_i=r_j<r_k$  or $r_i<r_j\leq r_k$. 

Note that if $\sigma$ is a linear order on $X$ and we delete one element from $X$, then any other element can have its rank reduced by at most 1.  Hence the richness of a domain can be reduced by at most 1 if we restrict the domain to $X\setminus a$, for some alternative $a$. 

Assume that we have the case $r_i=r_j<r_k$.  Here $r_k\geq \lfloor n/2 \rfloor+3$.  If we delete $i,j$ then the restricted domain has $n-2$ alternatives and by our induction assumption richness at most $\lfloor (n-2)/2 \rfloor+1=\lfloor n/2 \rfloor$.  But we also know that $k$ has  $r_k \geq \lfloor n/2 \rfloor+1$, and no alternative can have smaller range than this, which leads to a contradiction. 

If we have the case $r_i<r_j\leq r_k$ then we can reason in the same way, but now only deleting $i$. 
\end{proof}

As we will now show  both of the richness bounds we have derived are optimal. That is, there are domains with richness 2 and  and domains with richness $\lfloor n/2 \rfloor+1$ for every $n\geq 3$.

Let us say that an Arrow's single-peaked domain is a \emph{skewed single-peaked domain} if there exists an alternative $q$ such that every triple containing $q$ satisfies the never condition  $q$ not last. 
\begin{theorem}
     An Arrow's  single-peaked domain has richness 2 if an only if it is a skewed domain.
        
     The number of non-isomorphic skewed Arrow's single-peaked domains for $n$ alternatives is $a(n-1)$.
\end{theorem}
\begin{proof}
    Let $D$ be  an Arrow's single-peaked domain of richness 2.  Since the domain is not 3-rich there is an alternative $q$ which does not appear with rank 3.   Let us consider the never conditions on triples which contain alternatives $i$, $j$, and $q$. Since this is an Arrow's single-peaked domain each never condition is of the form $xN3$ for some $x$ in the triple, and since both $i$ and $j$ appear with rank 1 and 2 and $q$ does not get rank 3 in $D$, we must have that $q$ is not ranked last in any triple.  Hence $D$ is a skewed domain.

    By relabeling, any skewed domain $D$ is  isomorphic to a domain where the alternative $q$ is in fact $n$.  If the alternative $n$ is removed from $D$ then we are left with a general  Arrow's single-peaked domain on $n-1$ alternatives.

    Given an Arrow's single-peaked domain $D'$ on $n-1$ alternatives we can define a new domain $D$ on $n$ alternatives by adding for any triple $(i,j,n)$ the never condition $3N3$.  This domain will be a skewed single-peaked domain and the alternative $n$ can only appear with rank 1 or rank 2 in any order.   In fact we can partition $D$ into $D_1$ and $D_2$, where $D_1$ is the domain obtained from $D'$ by adding $n$ first to every order from $D'$, and $D_2$ is obtained by inserting $n$ with rank 2 in every order of $D'$.

    The domain $D$ will have richness 2, since $n$ only appears at rank 1 and 2, and all other alternatives appear  with rank 1 in $D_2$ and rank 2  in $D_1$. 
    
    Alternative $n$ is the only alternative which appears with rank 1  in $2^{n-1}$ of the orders, proving that every isomorphism between two skewed single-peaked domains will be determined by an isomorphism acting on the smaller domain $D'$, hence giving us $a(n-1)$ non-isomorphic skewed domains.  
\end{proof}
Following this proof we see that the domain $\mathcal{S}_n$ in which all never conditions are of the form $3N3$, or equivalently $1N3$ if we relabel the domain in order to make it unitary, will play a special role as the skewed domain  which remains skewed after deleting any set of alternatives. This is in some sense the most unbalanced of all Arrow's single-peaked domains.

Opposite to this we have the domain defined by assigning every triple the never condition $2N3$, and this is Black's single-peaked domain. Black's domain achieves the highest richness possible.
\begin{proposition}
    Black's single-peaked domain on $n$ alternatives has richness $\lfloor n/2 \rfloor+1$.
\end{proposition}
\begin{proof}
    Assume that we have the common axis $1,2,\ldots,n$. Here 1 and $n$ will be the two terminal alternatives.

    Let $\sigma$ be any order with $i$ ranked first.  Let $x$ be another alternative.   

    The lowest rank $x$ can have in any such $\sigma$ is equal to 1 plus the number of alternatives between $i$ and $x$ plus the number of alternatives on the other side of $x$. This sum is in fact equal to the distance between $x$ and the furthest endpoint on the common axis.

    For even $n$, letting $x$ be either of  the two alternatives $n/2$ and $n/2+1$ gives the smallest possible maximum for the two end-point distances, which is $n/2+1$.

    For odd $n$, taking $x$ to be $\lceil n/2 \rceil$ gives the smallest possible maximum for the two end-point distances, which is $\lfloor n/2\rfloor+1$.
\end{proof} 
However, as we can see in Table \ref{tab:rich} Black's single-peaked domain is far from alone in achieving the maximum possible richness. Unlike for richness 2 we do not yet have a natural description of the Arrow's single-peaked domains with maximum richness.
\begin{problem}
    Give a structural characterisation of the Arrow's single-peaked domains with maximum richness.
\end{problem}

\section{Arrow's single-peaked domains and axioms on irrelevant alternatives for different  voting systems}\label{sec:systems}
In \cite{dasgupta2008robustness} Dasgupta and Maskin investigated, for several different voting rules, which structural properties a domain must have in order for that voting rule to satisfy five natural axioms. They focused on voting rules for choosing a single winner, like a president,  and their main conclusion is that the simple majority rule has the largest well-behaved domains. As part of this work they gave  two characterisations of the domains on which the plurality rule and rank-order voting, or Borda's rule, satisfy Nash's version of independence of irrelevant alternatives.   Borda's rule has also been studied under other axioms. In \cite{barbie2006non}  the domains on which Borda's rule satisfies Arrow's version of independence of irrelevant alternatives together with another set  of axioms are studied, leading to a broader class of domains than those found under Dasgupta and Maskin's axiom.  We will now show how these domain classes connect to Arrow's single-peaked domains and discuss their size and richness.

First, let us recall the main voting rules discussed: The \emph{simple majority} rule chooses $x$ as the winner if for all other candidates $y$ more voters prefer $x$ than $y$; in \emph{rank-order voting} a candidate gets one point for each voter who ranks them first, two points for each who ranks them second, and so on. The winner is the candidate with the lowest number of points; The \emph{plurality rule},  chooses as winner the candidate which is ranked first by the largest number of candidates; in \emph{runoff voting} the plurality rule is applied if a candidate is ranked first by a majority, if that is not the case then it picks the majority winner among the subset of candidates which are ranked first by the largest number of voters.

\subsection{Plurality and Runoff voting}
Our first aim is to show that the plurality and runoff domains, under Dasgupta and Maskin's axioms,  are in fact exactly the Arrow's single-dipped domains.

The axiom which Dasgupta and Maskin focus much of their discussion on is independence of irrelevant alternatives.  Here one must be careful since over the years several non-equivalent axioms have appeared under that name, see \cite{ray73} for a survey of their interrelations.  The version used in \cite{dasgupta2008robustness}  is Nash's: 
\begin{definition}
    A  voting rule satisfies Nash's version of Independence of irrelevant alternatives (IIA) on a domain $D$, on a set of alternatives $A$, if  when $x\in A$ is ranked first by the rule for a profile and some set $x\notin B$  of alternatives is deleted, $x$ always remain ranked first.
\end{definition}
Simple majority satisfies IIA on all domains but both plurality and rank-order voting may not.

Dasgupta and Maskin introduced the following property, which we restate  in our terminology.
\begin{definition}
    A domain satisfies limited favoritism (LF) if  every triple  satisfies a never condition $qN1$ for some $q$ in the triple.
\end{definition}
They then showed: 
\begin{lemma}
    Both the plurality rule and runoff voting satisfies Nash's IIA on a domain $D$ if and only if $D$ satisfies LF.
\end{lemma}

Here now observe that the LF property is in fact a dual version of the definition of Arrow's single-peaked domains
\begin{proposition}
    A domain which satisfies LF  is an Arrow's single-dipped domain.
\end{proposition}
\begin{proof}
    Note that a domain which satisfies LF is a domain on which every triple satisfies a never condition of the form $qN1$. Taking the dual this gives a domain on which every triple has a never condition of the form $qN3$, which is the definition of an Arrow's single-peaked domain.
\end{proof}

As a consequence we have also provided a full enumeration of the domains on which either of the plurality rule and runoff voting satisfies Nash's version of IIA. 

Using this result we also find that the set of possible winners in a well-behaved domain for these voting rules is extremely sparse.
\begin{corollary}
    A domain which satisfies LF has at most two alternatives which are ranked first by some order in the domain.
\end{corollary}
This follows directly from the proposition and Raynaud's \cite{raynaud1981paradoxical} observation that an Arrow's single-peaked domain has at most two alternatives ranked last. 

Hence, both plurality and runoff voting have well-behaved domains which are as large as the classical single-peaked domains, but unlike these, the set of possible winners for these rules is as small as it can possibly be in a non-dictatorial domain. 

\subsection{Two versions of Borda domains}
Domains on which the Borda count is well behaved have been studied under several different sets of axioms.  In \cite{dasgupta2008robustness} a structural condition for domains on which the Borda count satisfies Nash's version  of IIA is derived.   In \cite{barbie2006non}  the authors instead consider a combination of Arrow's original version IIA with the Pareto principle for non-dictatorial domains.  

Our next aim is to expand on these results by identifying the largest domains for the Borda rule under each of the two versions of IIA. For Nash's IIA this leads to a family of non-maximal Arrow's single-peaked domains. For Arrow's IIA we identify the subclass of domains found in \cite{barbie2006non} which have maximum size.

Dasgupta and Maskin defined the following property, which we give a more concise definition of:
\begin{definition}
    A domain $D$ satisfies Quasi-agreement (QA)  if for any triple of alternatives $(i,j,k)$ the restriction of $D$ to this triple has a fixed point.
\end{definition}

In \cite{dasgupta2008robustness} they then showed that QA is equivalent to Nash's IIA for the Borda count.
\begin{proposition}
    Rank-order voting satisfies Nash's IIA  on a domain $D$ if and only if $D$ satisfies QA.     
\end{proposition}

We will now proceed to find the structure of the set of highest ranked alternatives in a QA domain, and the describe the QA domains of maximum size.
\begin{lemma}
    A QA-domain can have at most 2 alternatives ranked first.
\end{lemma}
\begin{proof}
    If there are three alternatives $(a,b,c)$ which can be ranked first then that triple of alternatives does not satisfy QA.
\end{proof}
This shows that, just like the plurality  and runoff rules,  the Borda count either fails to satisfy NAsh's IIA or has a very limited set of possible winners.

The following two lemmas are easy to prove so we leave their proofs out.
\begin{lemma}
    If 1 and 3 are ranked first in a unitary QA-domain then 2 is a fixed point.
    
    If 1 and $a>3$ are ranked fist in a unitary QA-domain $D$ then any order in $D$ restricted to the interval between 1 and $a$ is either the standard order on 1 to $a$ or its reverse.
\end{lemma}

\begin{lemma}
    If $b$ is a fixed point of a QA-domain $D_1$ then $D_2=D_1\setminus b$ is a QA-domain. The domain $D_2$ has the same size as $D_1$ and one fixed point fewer than $D_1$.

    If $D_1$ is a  QA-domain on $n-1$  alternatives and $D_1$ is obtained by adding a new alternative $n$ as a fixed point with rank $i$, for any $i$, then $D_2$ is also a QA-domain. 
\end{lemma}

For even $n$ let  $\mathcal{B}^1_n$ be the domain in which alternatives $2i-1$ and $2i$ are given either of rank $2i$ and $2i-1$ in very order. This domain has size $2^{n/2}$ and has no fixed points.
\begin{theorem}
    For even $n$ the unique unitary, maximum QA-domain is $\mathcal{B}^1_n$. 

    For odd $n$ every maximum QA-domain has at least one fixed point and there are  $n$ maximum unitary  QA-domains with exactly one fixed point. These domains have size $2^{(n-1)/2}$.
\end{theorem}
\begin{proof}
    The theorem is easily seen to be true for $n=3,4$. Likewise it is easy to see that domains of the stated sizes exist, either by taking $\mathcal{B}^1_n$ for even $n$, or by inserting a fixed point into $\mathcal{B}^1_{n-1}$ for odd $n$. 
    
    Let us assume that the theorem is true for less than $n$ alternatives and that $D$ is a maximum QA domain. 

    Assume that $D$ has a fixed point $p$. Then the size of $D$ is equal to that of the domain $D'$ obtained by deleting $q$. Since $D'$ has less than $n$ alternatives the theorem holds for $D'$.  If $n$ is odd this shows that $D$ has the size stated in the theorem. If $n$ is even we get that $D$ has size at most $2^{(n-2)/2}$. Hence  domain with a fixed point cannot be a maximum domain for even $n$.

    We can now assume that $D$ has two alternatives, 1 and $a$, which are ranked first. If  $a=2$ then we can delete 1 and 2 from each order an get a QA domain $D'$ for $n-2$ alternatives, which by induction satisfies our bound. Hence $D$ has size at most twice that of $D'$ and we are done.   If $a>2$ then we can delete $1,2\ldots,a$ from all orders and get a QA domain $D'$ for $n-a$ alternatives.  Our domain then has size at most twice that of $D'$, which is less than our bound if $n$ is even, or if $n$ is odd and $a>3$, and at most equal to our bound if $n$ is odd and $a=3$. 
\end{proof}

Here we  see that the maximum size of a QA domain is $2^{n/2}$, for even $n$,  while Arrow's single-peaked domains, and their duals, can have size $2^{n-1}$, giving a drastic difference in size.  However, for the dual Arrow's single-peaked domains one could claim that this difference is illusory, since both domain types have at most two possible winners.

Next we will  consider domains for Arrow's IIA, following the  characterisation from \cite{barbie2006non}.

\begin{definition}
    A  voting rule satisfies Arrow's version of Independence of irrelevant alternatives  on a domain $D$, on a set of alternatives $A$, if  when $x$ is ranked before $y$ by the rule for a profile and the voters change the ranks of some subset of alternatives which does not include $x$ and $y$,  $x$ always remain ranked before $y$ when the rule is applied to the new profile.
\end{definition}

Barbie et al. define the class of \emph{hierarchically cyclic} domains, and here we give their defintion adapted to unitary domains.
\begin{definition}
    A unitary domain is \emph{hierarchically cyclic} if the  alternatives can be partitioned into intervals $I_1,I_2,\ldots, I_t$ such that the following conditions are satisfied by the orders in $D$:
    \begin{enumerate}
        \item If $x\in I_i$ and $y\in I_j$ then $x$ is ranked higher than $y$ when $i<j$.  
        \item The restriction of $D$ to any interval $I_i$ is either a set of  cyclic shifts, or has size at most 2. 
    \end{enumerate}
\end{definition}
Note that the restriction of $D$ to an interval $I$ of size $k$ has size at most $k$, since there are exactly $k$ cyclic shifts.

Using this definition they then show:
\begin{theorem}
    The Borda count satisfies Arrow's IIA on a non-dictatorial domain $D$ if and only if $D$ is hierarchically cyclic.
\end{theorem}
As noted in \cite{barbie2006non} a hierarchically cyclic domain is minimally rich only when there is a single interval of length $n$ in the partition, leading to a domain of size at most $n$.

It follows from the definition that the maximum domain size for a given partition into intervals is $\prod_{i=1}^t k_i$, where $k_i=|I_i|$.  Finding the maximum possible domain size becomes a combinatorial problem: Which partition of an integer $n$ maximizes the product of the parts?  
\begin{theorem}
    The maximum  size of  a hierarchically cyclic domain is $3^{n/3}$, $4\times 3^{(n-4)/3}$, or $2\times 3^{(n-2)/3}$, for $n$ congruent to $0, 1$ and $2$ modulo 3 respectively. 

    A maximum domain has 3 winners for $n$ congruent to 0, 3 or 4 winners for $n$ congruent to 1 and 2 or 3 winners for $n$ congruent to 2 modulo 3.    
\end{theorem}
\begin{proof}
    The optimal partition of $n$ will consist of multiple parts of size 3, and a single part of size 4 or 2 when $n$ is congruent to 1 and 2 modulo 3 respectively.

    To see this we note that the combinatorial problem i equivalent to maximizing $\sum_i \log{k_i}=\sum_k n_k \log{k}$, where $n_k$ is the number of parts of size $k$.  The count $n_k$ is bounded from above by $n/k$  and the summand is bounded by $n/k\log{k}$. Here we find that  for integer $1/k \log{k}$ is maximized at $k=3$, and is -1  for $k=1$. Hence the sum is maximized when the number of parts of size 3 is maximized, under the condition that no part of size 1 is used.  
    
    For $n$ congruent to 0 modulo 3 we get a partition using $n/3$  parts of size 3, giving us 3 possible winners.

    For $n$ congruent to 1 modulo 3 we get a partition using $(n-4)/3$  parts of size 3, and one part of size 4,  giving us 3 or 4 possible winners depending on whether the part of size 4 is used as the first interval or not.

    For $n$ congruent to 2 modulo 3 we get a partition using $(n-2)/3$  parts of size 3 and one part of size 2, giving us 2 or 3 possible winners, depending on whether the part of size 2 is used as the first interval or not.
\end{proof}

Here we find that the domains on which the Borda count satisfies Arrow's version of IIA are larger than  for the Nash version.  The growth rate here is $3^{1/3}\approx 1.44$ compared to $2^{1/2}\approx 1.41$ for Nash.  However, the size is still much smaller than for Arrow's single-peaked domains, and the number of possible winners is again bounded by a constant for the maximum domain. 

We can of course prescribe the number of possible winners in a hierarchically cyclic domain by taking the first interval to have length $p$ and then optimizing the size with respect to the other intervals. Ignoring modulo effects this gives a domain of size $q 3^{(n-q)/3}$. Up to $n=12$ this makes it possible to have $q\geq n/2$ while having a domain larger than the domains for Nash's IIA. However, as $n$ grows  the maximum $q$ which gives a larger domain for Arrow's IIA shrinks down to $q=cn$, where $c=1-\frac{\log(8)}{\log(9)}\approx 0.0536$

\section{Discussion}
In this paper we have first verified and extended the existing enumerations of non-isomorphic Arrow's single-peaked domains.  Following up this we showed that despite being similar  in terms of many properties  the different Arrow's single-peaked domains vary considerably in terms of richness for the high ranked alternatives, with Black's single-peaked domain providing an example of maximum richness for small $n$.   Our results show that in addition to Black's single-peaked domain there are other Arrow's single-peaked domains which provide the same amount of richness, in terms of high ranked alternatives.

After this computational study we connected Arrow's single-peaked domains to the study of various voting rules for choosing a single winner. In particular we showed that the domains on which plurality and runoff voting  satisfy Nash's version of independence of irrelevant alternatives are exactly the duals of the Arrow's single-peaked domains. A consequence of this is that these voting systems fail to satisfy Nash's IIA  in any domain with more than two first ranked alternatives.

Finally we extended this comparison to include Borda's rule. This voting rule is well-behaved on different domains depending on whether Arrow's or Nash's version of IIA is applied, but in each case the domains are small compared to the Arrow's single-peaked domains, and have a limited number of possible winners.

Our results on domains in which the different versions of IIA are satisfied show that these domains are small, compared to Arrow's single-peaked domains, and Condorcet domains more generally.  This could be seen as either a weakness of the voting rules concerned or a too strong constraint from the axioms applied.  These axioms have been debated for a long time, see e.g. \cite{patty19} for a recent survey and defense of Arrow's IIA, and independently of the appealing consequences of the axiom(s) it is known that very few social choice rules satisfy it unless the domain is restricted. Our results demonstrate just how tight those domain restrictions have to be.  

On the other hand,  Maskin \cite{maskin2023borda} has recently proposed a new version of IIA which is satisfied by Borda's rule without the need for any domain restriction, though in a setting with a continuum of voters. So, after more than 70 years of study this is still an area with new possibilities to consider.

\section*{Acknowledgement}
We would like to thank Alexander Karpov and Arkadii Slinko for their comments on the manuscript of this paper. Bei Zhou was funded by the China Scholarship Council (CSC).  This research utilised Queen Mary's Apocrita HPC facility, supported by QMUL Research-IT. This research was conducted using the resources of High Performance Computing Center North (HPC2N).


\ifx\undefined\bysame
\newcommand{\bysame}{\hskip.3em \leavevmode\rule[.5ex]{3em}{.3pt}\hskip0.5em}
\fi
\bibliographystyle{econ}

\end{document}